\documentclass[%
 reprint,
 amsmath,amssymb,
 aps,
]{revtex4-2}

\usepackage{graphicx}% Include figure files
\usepackage{dcolumn,mathrsfs,amssymb}% Align table columns on decimal point
\usepackage{bm}% bold math
\newtheorem{theorem}{Theorem}

\newtheorem{proposition}[theorem]{Proposition}
\newtheorem{remark}[theorem]{Remark}

\newenvironment{proof}[1][Proof]{\noindent\textbf{#1.} }{\ \rule{0.5em}{0.5em}}

\begin{document}

\preprint{APS/123-QED}

\title{The Universe as a Detector:\\ A Quantum Filtering Formulation of the Di\'osi-Penrose Model}% Force line breaks with \\

\author{John E. Gough}
\email{jug@aber.ac.uk}
\author{Dylon Rees}%
 \email{dyr6@aber.ac.uk}
\affiliation{%
 Department of Physics, Aberystwyth University, Ceredigion, Wales, UK. %\textbackslash\textbackslash
}%



\date{\today}% It is always \today, today,

\begin{abstract}
We consider the Di\'osi-Penrose problem but rather than postulating background gravitational fluctuations, we instead consider the quantum filter that arises from space-time homodyning the continuum of output quadrature described in the open quantum stochastic model presented here. This is described by a quantum Kushner-Stratonovich equation, typical of the form appearing in continuous-time collapse of the wave-function models in Quantum Decoherence Theory.
\end{abstract}

\maketitle

%\tableofcontents

\section{Introduction}
The modern theory of filtering is due to Rudolf Kalman and has found widespread applications to signal processing ever since it was used to track the Apollo spacecraft on all six NASA missions to the Moon. The essential problem is to estimate the state of a system with a known dynamics, potentially subject to additional dynamical noise, when we have (partial information which may itself be subject to further observational noise. The least squares estimator is known as the \textit{filter} as it effectively filters out the dynamical and observational noise. What is less well-known is that there are well-defined quantum analogues applicable to quantum open systems undergoing continual measurement.

Quantum theory does not so much tell us how the microscopic world works but rather how it appears to large macroscopic observers (ourselves included) when we attempt to measure it. An early example of this issue is the Mott problem which asks the question as to why emitted $\alpha$-particles from a radioactive source in a Wilson cloud chamber lead to linear tracks when they are described quantum mechanically by a spherically symmetric wave-function. Mott's solution brought forward the realization that the cloud chamber was in fact acting as a detector \cite{Mott} with the $\alpha$-particle interacting with individual molecules of steam (a detection event that led to a condensed droplet). This was further developed by Joos and Zeh \cite{Joos_Zeh} within the framework of quantum decoherence.

An alternative is to treat this as a filtering problem: the position of the $\alpha$-particle is unknown but we learn about it through indirect observations (the droplets) with both the dynamics and the observations being noisy. To be sure, the physics underlying the problem of tracking a spacecraft using radar is fundamentally different from that of tracking a quantum particle, however, the basic statistical problem translates surprisingly naturally from the classical to quantum probability setting \cite{Slava,BvH}. This has been applied extensively to quantum optics. The Stratonovich-Kushner equation from nonlinear filtering gives the least squares estimate for a dynamic quantity conditional on the (partial) measurements obtained up to that time. And this turns out to have a quantum analogue, called the quantum filter: this is described by a stochastic master equation, typical of the form appearing in continuous-time collapse of the wave-function models in quantum decoherence theory.

We wish to consider the specific model for collapse of the wave-function due to Newtonian gravity introduced by Di\'osi \cite{Diosi87,Diosi89}, which was also qualitatively considered by Penrose \cite{Penrose}. Di\'osi considers the following stochastic Schr\"{o}dinger equation for a quantum system (consisting of several particles with position observables $\hat x_i$) experiencing decoherence from a classical field leading to a collapse in the position:
\begin{eqnarray}
    d|\psi_t\rangle 
        &=& 
        \sqrt{\kappa}\int d^3x \,  \bigg( \hat \mu(x)- \langle{\hat \mu(x)\rangle_t}\bigg)|\psi_t \rangle\, dW(x,t) \nonumber \\
       &&+        
        \bigg\{\frac{1}{i\hbar}\hat H - \frac{\kappa}{2} \int\int \frac{dx^3 d^3 y}{|x-y|}  \bigg( \hat \mu(x)^\ast- \langle{\hat \mu(x)^\ast\rangle_t}\bigg) \nonumber \\
        &&\times
    \bigg(\hat \mu(y)- \langle{\hat \mu(y)\rangle_t}\bigg) \bigg\} |\psi_t \rangle\, dt .
    \label{eq:DiosiSME}
\end{eqnarray}
Here $W(x,t)$ is a classical random noise field with mean zero and correlation
\begin{eqnarray}
    \mathbb{E} [W(x,t)W(y,s)]= \frac{1}{|x-y|}\delta (t-s).
    \label{eq:DiosiCorr}
\end{eqnarray}
There is a coupling constant $\kappa$ identified as $G/\hbar$ and one has an observable-valued function of the spatial variable $x$:
\begin{eqnarray}
    \hat \mu(x)=\hat \mu(x)^\ast =\sum_i m_i \varphi (x-\hat x_i) .
\end{eqnarray}
The function $\varphi$ is to be a real-valued approximation to a delta function. As usual $\hat H$ is the system's Hamiltonian. The equation is nonlinear in the state $| \psi_t \rangle$ as it involves the averages $\langle \hat \mu (x) \rangle_t = \langle  \psi_t | \hat \mu(x) | \psi_t \rangle$.

\bigskip

Rather than postulating background gravitational fluctuations, we instead consider how Di\'osi's stochastic master equation may arise as a quantum Kushner-Stratonovich equation for the conditioned state of a particle.

%%%%%%%%%%%%%%%%%%%%%%%%%%%%%%%%%%%%%%%%%%%%%%%
\subsection{Preliminaries}
Di\'osi's stochastic master equation (\ref{eq:DiosiSME}) gives an unraveling of the associated master equation
\begin{eqnarray}
    \frac{d\hat\rho(t)}{dt} &=& \frac{1}{i\hbar} [\hat H, \hat\rho(t)]  +\kappa\int\int \frac{dx\, dy}{|x-y|}\nonumber \\
   &&
    \times\bigg( \hat \mu(x) \hat\rho (t) \hat \mu(y)-\frac{1}{2} \big\{ \hat \rho(t) , \hat \mu(x)\hat \mu(y) \big\} \bigg).
\end{eqnarray}

For our purposes, it is more convenient to formulate the problem in the Heisenberg picture: here $\frac{d}{dt} \langle \hat X \rangle_t = \langle\mathcal{L} (\hat X) \rangle_t$ where the Di\'osi-Penrose Lindbladian is
\begin{eqnarray}
    \mathcal{L} (\hat X) &=&
    \frac{1}{i\hbar}[ \hat X ,\hat H ]
    + \frac{ \kappa }{2}\int\int \frac{dx^3 d^3y}{|x-y|} \nonumber \\
    &&\times
    \bigg( [ \hat \mu(x), \hat X ] \hat \mu(y) + \hat \mu(x) [ \hat X , \hat \mu(y) ]\bigg)  .
    \label{eq:DP_Lindbladian}
\end{eqnarray}

The super-operator Linbladian (\ref{eq:DP_Lindbladian}) is indeed a genuine Lindblad generator. To see this, we compute the dissipation $\mathscr{D}_{\mathcal{L}} (\hat X , \hat Y) = \mathcal{L}(\hat X \hat Y)-\mathcal{L}(\hat X)\hat Y-\hat X \mathcal{L}(\hat Y)$ and here we see that
\begin{eqnarray}
    \mathscr{D}_{\mathcal{L}} (\hat X ^\ast, \hat X) =  \kappa \int\int \frac{d^3x d^3y}{|x-y|}
     [ \hat \mu(x), \hat X ]^\ast [ \hat \mu(x), \hat X ] \ge 0.
\end{eqnarray}
This ensures that the semi-group generated by $\mathcal{L}$ is positive and, indeed, completely positive. Note that this formula for the dissipation holds true even if we drop the condition that the $\hat \mu(x)$ form a commutative family.

Our first order of business is to take the Di\'osi-Penrose Linbladian (\ref{eq:DP_Lindbladian}) and rewrite it in a diagonal form. Before going into the technicalities, we describe the problem in a discrete setting. Let us suppose that we have a Linbladian of the following form of a finite sum:
\begin{eqnarray}
    \mathcal{L} (\hat X) =\frac{1}{2} \sum_{ij} g^{ij} \bigg( [\hat M_i^\ast ,\hat X ]\hat M_j
    + \hat M_i^\ast [\hat X ,\hat M_j ] \bigg).
\end{eqnarray}
The $g^{ij}$ form the components of a positive definite matrix and we take the spectral expansion of this matrix to be
\begin{eqnarray}
    g^{ij}= \sum_a \lambda_a \, e^i(a) ^\ast \, e^j(a) .
    \label{eq:discrete_diag}
\end{eqnarray}
The Lindbladian may then be recast in the traditional form \cite{Lindblad}
\begin{eqnarray}
    \mathcal{L} (\hat X) =\frac{1}{2} \sum_{a} 
    \bigg( [\hat L_a^\ast ,\hat X ]\hat L_a
    + \hat L_a^\ast [\hat X ,\hat L_a ] \bigg)
\end{eqnarray}
where the collapse operators are
\begin{eqnarray}
    \hat L_a = \sqrt{\lambda_a }\, \sum_j e^j (a) \, \hat M_j .
\end{eqnarray}

At this stage, the route to developing a filtering model with these collapse operators is routine. We begin by setting up a dilation of the CP semi-group generated by $\mathcal{L}$ using the Hudson-Parthasarathy  quantum stochastic calculus \cite{HP84,Partha}; then we identify the input-output relations; finally, we fix an appropriate measurement scheme for the output fields \cite{GC85}.

The main technical challenge will be handling the continuous collection of collapse operators rather than a discrete sum.

%%%%%%%%%%%%%%%%%%%%%%%%%%%%%%%%%%%
\subsection{Diagonalizing the Poisson Equation Green's Function}
The Newtonian potential $\varphi$ generated by a mass density $\mu$ in $\mathbb{R}^3$ satisfies Poisson's equation 
\begin{eqnarray}
    \nabla^2 \varphi = - 4 \pi G \, \mu
\end{eqnarray}
where $G$ is the universal gravitational constant.
The potential is given by $\varphi (x)= \int d^3 y \, g(x,y)\, \mu (y) $ where the Green's function is
\begin{eqnarray}
    g(x,y) = \frac{ G}{|x-y|}.
\end{eqnarray}
The calculation of the Green's function is well-known, but worth recalling as it essentially gives the desired analogue to the decomposition (\ref{eq:discrete_diag}):
\begin{eqnarray}
    g(x,y) = \int d^3 \, \tilde{g} (k)  \, e_k(x)^\ast e_k(y)
\end{eqnarray}
The relevant eigen-functions are the plane waves $e_k$ parameterized by wave-number $k \in \mathbb{R}^3$:
\begin{eqnarray}
    e_k (x)= \frac{1}{(2\pi )^{3/2}}\, e^{ik \cdot x} .
    \label{eq:e-vector}
\end{eqnarray}
Fourier transform pairs are then given by
\begin{eqnarray}
    f (x) &=& \int d^3k\, \tilde{f} (k) e_k (x) \nonumber \\
    \tilde{f} (k)&=& \int d^3 x \, f (x)e_k (x)^\ast .
\end{eqnarray}
The standard completeness relations $    \int d^3 k \, e_k (x)^\ast e_k (x')$ $ =\delta^3 (x-x')$, and $  \int d^3 x \, e_k (x)^\ast e_{k'} (x) = \delta^3 (k-k')$ hold.
The Fourier transformed Poisson's equation is $-|k|^2 \tilde{\varphi}(k) = 4 \pi G \, \tilde{\mu} (k)$, or $ \tilde{\varphi} = \tilde{g}(k) \, \tilde{\mu}(k)$ with
\begin{eqnarray}
    \tilde{g}(k) = 4\pi G / |k|^2
\end{eqnarray}
In detail, we have $  g(x,y) =4\pi  G \int\frac{d^3 k}{|k|^2} \, \frac{1}{(2 \pi )^2} e^{ik \cdot (x-y)}$. The $k$-integration can be simplified by choosing the spherical polar coordinates $(\kappa , \theta , \phi )$ with $\kappa = |k|$ and  $\theta$ the angle between the vectors $k$ and $x-y$. The integral then simplifies to
 \begin{eqnarray}
    g(x,y) &=&\frac{G}{\pi}  \int_0^\infty \kappa^2 d\kappa \int_0^\pi \sin \theta d\theta \, \frac{1}{\kappa^2}e^{i\kappa |x-y|\cos \theta} \nonumber \\
    &=&
    \frac{2G}{\pi}  \int_0^\infty  d\kappa \, \text{sinc}  (\kappa |x-y|) \nonumber \\
    &=&
    \frac{G}{|x-y|}
\end{eqnarray}
where $\text{sinc} (u)= \frac{\sin u}{u}$ and we use the integral identity $\int_0^\infty \text{sinc} (u) \, du = \frac{\pi}{2}$.

\begin{proposition}
The Green's function $g$ admits a convolutional \lq\lq square root\rq\rq\, $\gamma$ such that 
\begin{eqnarray}
    \int d^3x' \, \gamma (x,x')\gamma (x' , x'') = g(x-x'') .
\end{eqnarray}
This is explicitly given by
\begin{eqnarray}
    \gamma (x,y) =
    \sqrt{\frac{G}{\pi^3}} \, \frac{1}{|x-y|^2}.
\end{eqnarray}
\end{proposition}
\begin{proof}
The convolutional square-root of the Green's function is then given by
\begin{eqnarray}
    \gamma (x,y)&=&\int dk^3 \, \sqrt{\tilde{g} (k)}  \, e_k(x)^\ast e_k(x') \nonumber \\
    &=& \sqrt{4\pi \, G}\int \frac{d^3k}{|k|} \, e_k(x)^\ast e_k(x') \nonumber \\
    &=&
    \sqrt{\frac{G}{\pi^3}}  \int_0^\infty  dk \, \frac{\sin (k |x-y|)}{|x-y|},
\end{eqnarray}
but this must be interpreted in the distributional sense. We may use the Heitler function $\int_0^\infty e^{ixk}dk = \pi \delta (x) +i \text{PV} \frac{1}{x}$ to write $\int_0^\infty \sin (xk) \, dk =  \text{PV} \frac{1}{x}$, where we encounter principal value terms, and this leads to the proposed expression.

To see that this has the desired property, we note that 
\begin{eqnarray*}
    \int d^3x' \, \gamma (x,x')\gamma (x' , y) = \frac{G}{\pi^3} \int d^3 x' \frac{1}{|x-x'|^2} \frac{1}{|x'-y|^2} .
\end{eqnarray*}
Let us translate the integration variable from $x'$ to $x'-x$ which we write in spherical polar form $(\rho,\theta, \phi)$ where $\rho=|x'-x|$ and $\theta $  is the angle of $x'-x$ relative to $x-x''$. Setting $r=|x-x''|$, we have
\begin{eqnarray}
   && \frac{G}{\pi^3} \text{PV} \, \int_0^\infty \rho^2d\rho\int_0^\pi d\theta \int_0^{2\pi}d\phi \,
    \frac{1}{\rho^2}\frac{1}{\rho^2+r^2-2r\rho \cos \theta} \nonumber \\
    &&=\frac{2G}{\pi^2}\int_0^\infty d\rho \int_{-1}^1 d u \,
    \frac{1}{\rho^2+r^2-2r\rho u}
    \nonumber \\
    && = \frac{2G}{\pi^2}\int_0^\infty d \rho \,
    \frac{1}{r\rho} \ln \left| \frac{r+\rho}{r-\rho}\right| .
\end{eqnarray}
The integrand here is, in fact, integrable. Moreover, the integrals $\int_0^r d \rho \,
    \frac{1}{\rho} \ln \left| \frac{r+\rho}{r-\rho}\right| $ and $\int_r^\infty d \rho \,
    \frac{1}{\rho} \ln \left| \frac{r+\rho}{r-\rho}\right| $
    are both finite and equal to $\frac{\pi^2}{4}$. Therefore, the principal value is not needed and the answer is given by $G/r$ which is the desired Green's function.
\end{proof}

\begin{remark}
    In particular, the self-energy associated with the mass distribution is
\begin{eqnarray}
    \Gamma = \int d^3x d^3 y \,  \mu (x) g(x,y) \mu(y) .
\end{eqnarray}
For a system consisting of point masses $m_i$ at positions $x_i$, we should set $\mu (x)= \sum_i m_i \, \delta (x - x_i)$. In practice, we may mollify this to $\mu (x)= \sum_i m_i \, \delta^{(\epsilon )} (x - x_i)$ where $ \delta^{(\epsilon )}$ is a regular function approximating the singular delta-function in the limit $\epsilon \to 0$.
\end{remark}

\begin{remark}
    It is tempting to try and construct a reproducing kernel Hilbert space \cite{RKHS} for the Green's function $g$ as it is clear that the self-energy is positive for all $\mu$ (even signed). However, the Green's function is singular with $g(x,x)=+\infty$. Nevertheless, The functions $\gamma_x (\cdot )= \gamma (x, \cdot)$ play the role of representers while $g(x,y) = \int d^3 \, \tilde{g} (k)  \, e_k(x)^\ast e_k(y)$ plays the role of a Mercer decomposition \cite{RKHS}.

    Indeed, let us denote by $L^2(\mathbb{R}^3)$ be the Hilbert space of square-integrable functions with the usual inner product, and  by $L^2(\mathbb{R}^3,g)$ be the Hilbert space with inner product $\langle \phi,\psi \rangle_g =\int d^3x d^3 y \,  \phi^\ast (x) g(x,y) \psi(y)  $, then for each $\psi \in L^2(\mathbb{R}^3,g)$ we may define
    \begin{eqnarray}
        \check{\psi} (x) =\langle \gamma _x , \psi\rangle
    \end{eqnarray}
    so that we have
    \begin{eqnarray}
        \langle \check{\phi} , \check{\psi} \rangle 
        =\langle \phi , \psi \rangle_g .
    \end{eqnarray}
\end{remark}
%%%%%%%%%%%%%%%%%%%%%%%%%%%%%%%%%%%%%%%%%%%%%%%%%
%%%%%%%%%%%%%%%%%%%%%%%%%%%%%%%%%%%%%%%%%%%%%%%%%
%%%%%%%%%%%%%%%%%%%%%%%%%%%%%%%%%%%%%%%%%%%%%%%%%
%%%%%%%%%%%%%%%%%%%%%%%%%%%%%%%%%%%%%%%%%%%%%%%%%
%%%%%%%%%%%%%%%%%%%%%%%%%%%%%%%%%%%%%%%%%%%%%%%%%
%%%%%%%%%%%%%%%%%%%%%%%%%%%%%%%%%%%%%%%%%%%%%%%%%

\section{Quantum Open System Model}
A quantum model for this is obtained by fixing a Hilbert space $\mathfrak{h}_{\text{sys.}}$ for the system of massive quantum particles and introducing a scalar bosonic quantum field in (Newtonian) space-time with underlying Fock space $\mathfrak{F}$.

The mass distribution $\mu (x)$ now becomes a family of commuting observables on $\mathfrak{h}_{\text{sys.}}$: for instance, we could consider the point position $x_i$ being now replaced by the position observables $\widehat{x}_i$ for the masses.

The field is described in terms of operator annihilator densities $\widehat{b}(x,t)$ satisfying the singular commutation relations
\begin{eqnarray}
    [\widehat{b}(x,t) , \widehat{b}(y,t') ^\ast]= \delta^3(x-y) \, \delta (t-t') .
\end{eqnarray}
Denoting the Fock vacuum as $| \text{Vac}\rangle$, we have that $\widehat{b}(x,t) \, | \text{Vac}\rangle =0$. 

\bigskip

In the following, it is convenient to introduce modified input processes defined as 
\begin{eqnarray}
    \widehat{\beta} (x,t) = \int d^3x' \, \gamma (x,x') \, \widehat{b}(x', t).
\end{eqnarray}
Here we have the commutation relations; 
\begin{eqnarray}
    [\widehat{\beta} (x,t) ,\widehat{\beta} (x', t')^\ast] = g(x,x') \, \delta (t-t').    
\end{eqnarray}
We also consider their quadratures:
\begin{eqnarray}
    \widehat{\phi} (x,t) &=& \widehat{\beta} (x,t)- \widehat{\beta} (x,t)^\ast \nonumber \\
    \widehat{\pi} (x,t) &=&
    \frac{1}{i} \big( \widehat{\beta} (x,t)- \widehat{\beta} (x,t)^\ast \big) .
\end{eqnarray}
The quadratures are separately Gaussian fields for the choice of vacuum state with mean-zero and two-point functions:
\begin{eqnarray}
    \langle \text{Vac} |\, \widehat{\phi} (x,t) \, \widehat{\phi} (y,t') |\text{Vac} \rangle = g(x,y) \, \delta (t-t') , \, \text{etc}.
\end{eqnarray}
However, we note that they do not commute among themselves:
\begin{eqnarray}
    [\widehat{\phi} (x,t) , \widehat{\pi} (x',t')] = 2i \, g (x-x') \, \delta (t-t').
\end{eqnarray}

On the joint Hilbert space 
$\mathfrak{h}_{\text{sys.}}\otimes \mathfrak{F}$, we consider the Hamiltonian
\begin{eqnarray}
    \widehat{\Upsilon} (t) =\widehat{H} \otimes \hat I+ 
    \int d^3x\, \widehat{\mu} (x) \otimes \widehat{\pi} (x,t).
    \label{eq:Upsilon}
\end{eqnarray}

\bigskip

It is again convenient to work in the Fourier domain $(x\to k)$. Let us introduce the densities
\begin{eqnarray}
    \tilde b (k,t) = \int d^3 x \, e_k (x)^\ast \, \widehat{b}(x,t)
\end{eqnarray}
in which case $[\tilde b( k, t) , \tilde b (k',t')^\ast ] = \delta^3 (k-k') \, \delta (t-t') $. (We note that $ \widehat{\beta} (x,t) =\sqrt{4\pi \, G} \int \frac{d^3 k}{|k|} \, e_k (x) \, \tilde b(k,t) $.) The Hamiltonian (\ref{eq:Upsilon}) may be rewritten as
\begin{eqnarray}
    &&-i\widehat{\Upsilon} (t)=-i\hat H \otimes \hat{I}\nonumber \\
    &&    + 
    \int d^3k\, \bigg( \widehat{L} (k) \otimes \tilde b (k,t)^\ast
    -\widehat{L} (k)^\ast  \otimes \tilde b (k,t) \bigg) 
\end{eqnarray}
where
\begin{eqnarray}
    \widehat{L} (k) = \frac{\sqrt{4 \pi \, G}}{|k|}\int d^3x \, e_k (x) \, \widehat{\mu} (x) .
    \label{eq:L_Fourier}
\end{eqnarray}

We remark that $\widehat{L}(k)= \frac{\sqrt{4 \pi \, G}}{|k|} \sum_i \frac{m_i}{(2\pi )^{3/2}}e^{ik \cdot \widehat{x}_i}$, for idealized point mass distributions.

%%%%%%%%%%%%%%%%%%%%%%%%%%%%%%%%%%%%%%%%%%%%%%%%%
\subsection{Quantum Stochastic Calculus}
We now need to set up a quantum stochastic calculus tailored to this problem. This will be the Hudson-Parthasarathy theory of Fock space processes, but with the color space being $\mathfrak{K}=L^2 (\mathbb{R}^3, \tilde{g} (k) \, d^3k)$. Note that we could equivalently have taken the color space to consist of functions of the space variables, but this will be $L^2(\mathbb{R}^3, g)$ and we have the unitary map from $x\mapsto \psi (x)$ to $k\mapsto \sqrt{\tilde{g}(k)} \, \tilde{\psi} (k)$.

For each fixed $k$, we define the process
\begin{eqnarray}
    \widetilde B (k, t) = \int_0^t \tilde b (k ,t') \, dt'
\end{eqnarray}
(Properly speaking, we should be working with processes of the form $\int_0^t dt' \int d^3 k $ $  \tilde f (k,t')^\ast \, b(k,t') $
where $\int_0^t dt' \int d^3 k \,  |\tilde f (k,t')|^2 < \infty$, which are well-defined operators on the Fock space, but the meaning should be clear.)

The differentials $\widetilde B (k, dt)$ are to be understood as corresponding to future-pointing increments $\widetilde B (k, t+dt)-\widetilde B (k, t)$. The quantum Ito table will be
\begin{eqnarray}
    \begin{array}{c|cc}
        \times  &  \widetilde B (k',dt) & \widetilde B (k',dt)^\ast  \\[4pt]  \hline
         &  \\
       \widetilde  B(k,dt) &  0 & \delta^3 (k-k')  \, dt \\
        \widetilde B(k,dt)^\ast & 0 & 0
    \end{array}
    .
\end{eqnarray}

%%%%%%%%%%%%%%%%%%%%%%%%%%%%%%%%%%%%%%%%%%%%%%%%%
\subsection{Unitary Dynamics}
The unitary evolution is given by the family of operators $\hat U(t)$ on $\mathfrak{h}_{\text{sys.}}\otimes \mathfrak{F}$ satisfying 
\begin{eqnarray}
    \frac{d}{dt}\hat U(t) = \frac{1}{i\hbar} \hat\Upsilon (t) \,\hat U(t),
    \label{eq:Schrodinger}
\end{eqnarray}
with $\hat U(0)=I$.

In its present form, the annihilator densities $\tilde b (k,t)$ appear to the left of $\hat U(t) $, however, it is easy to carry out the Wick ordering procedure. 

\begin{proposition}
    The equation (\ref{eq:Schrodinger}) is equivalent to the quantum stochastic differential equation
    \begin{gather}
    d\hat U(t) = - \bigg( \frac{i}{\hbar}\hat H+\frac{1}{2\hbar}\hat\Gamma \bigg) \otimes dt \, U(t)  \nonumber \\
    +\int d^3 k \bigg\{ 
    \hat L(k) \otimes \hat B(k,dt)^\ast -
    \hat L(k)^\ast \otimes \hat B(k,dt) \bigg\} \hat U(t).
\end{gather}
\end{proposition}
\begin{proof}
    The Wick ordering is straightforward and we pick up the quantum Ito correction $ -\frac{1}{2\hbar}\int d^3k \, \hat L(k)^\ast\hat L(k)$ and by substituting in (\ref{eq:L_Fourier}) we see that this reduces to the gravitational self-energy $-\frac{1}{2\hbar} \hat \Gamma$.
\end{proof}

\begin{proposition}
    The system observables evolve as $\hat X\to j_t(\hat X)= \hat U(t)^\ast (\hat X\otimes \hat I)\hat U(t)$ and we have the associated equations of motion:
    \begin{eqnarray}
        dj_t (\hat X) &=& j_t (\mathcal{L}\hat X)
           +\int d^3 k \bigg\{ 
    j_t ([\hat X,\hat L(k)] \otimes \hat B(k,dt)^\ast %\nonumber \\
    \nonumber\\
    && \qquad \qquad
    -
    j_t [\hat X,\hat L(k)^\ast ]\otimes \hat B(k,dt) \bigg\}
    \end{eqnarray}
    where the Lindblad generator is
    \begin{gather}
        \mathcal{L} \hat X = \frac{1}{i\hbar}[\hat X,\hat H] \nonumber \\
        +\frac{1}{2\hbar} \int d^3 k \bigg\{ [\hat L(k)^\ast, \hat X] \hat L(k) +
    \hat L(k)^\ast [\hat X, \hat L(k)]  \bigg\}        .
    \end{gather}
\end{proposition}
The proof is an application of the quantum Ito calculus. 
Substituting in (\ref{eq:L_Fourier}) for the collapse operators, we find that the generator may also be written in terms of the spatial distribution of the masses and that it agrees with (\ref{eq:DP_Lindbladian}).
\begin{proposition}
    The input processes are defined to be $\tilde b_{\text{in}} (k,t)= I \otimes \tilde b (k,t)$ while the output fields are 
    \begin{eqnarray}
     \tilde b_{\text{out}} (k,t)= U(t)^\ast \tilde b_{\text{in}} (k,t) U(t).   
    \end{eqnarray}
    The input-output relation is then
    \begin{eqnarray}
        \tilde b_{\text{out}} (k,t) =\tilde b_{\text{in}} (k,t) +j_t \big( \hat L(k) \big) .
    \end{eqnarray}
\end{proposition}
This is again a routine application of the quantum Ito calculus. For our purposes, it is more convenient to write this as
\begin{eqnarray}
    \widetilde{B}_{\text{out}} (k,dt) =\widetilde{B}_{\text{in}} (k,dt) +j_t \big( \hat L(k) \big) dt.
\end{eqnarray}
%%%%%%%%%%%%%%%%%%%%%%%%%%%%%%%%%%%%%%%%%%%%
%%%%%%%%%%%%%%%%%%%%%%%%%%%%%%%%%%%%%%%%%%%%
%%%%%%%%%%%%%%%%%%%%%%%%%%%%%%%%%%%%%%%%%%%%
%%%%%%%%%%%%%%%%%%%%%%%%%%%%%%%%%%%%%%%%%%%%
%%%%%%%%%%%%%%%%%%%%%%%%%%%%%%%%%%%%%%%%%%%%
\section{Derivation of the Quantum Filter}
We consider the problem of continuously monitoring the following quadratures of the output fields for each $k$:
\begin{eqnarray}
    \hat Y(k,t) = \widetilde{ B} _{\text{out}} (k,t) +\widetilde{ B} _{\text{out}} (k,t)^\ast .
\end{eqnarray}
At time $t$, we will have measured the observables $\hat Y(k,s)$ where $k \in \mathbb{R}^3$ and $0\le s \le t$. We denote the corresponding von Neumann algebra generated this way as $\mathfrak{Y}_t$. In this way, we obtain a filtration of nested, commutative von Neumann algebras $\{ \mathfrak{Y}_t: t \ge 0 \}$.

Our goal is to find the quantum conditional expectation $\pi_t (\hat X)$ for the Heisenberg picture value $j_t (\hat X )$ of any system observable on to $\mathfrak{Y}_t$.

\begin{theorem}[The Quantum Filtering Equation]
The filter satisfies the stochastic differential equation
    \begin{eqnarray}
        d\pi_t (\hat X ) &=& \pi_t (\mathcal{L} \hat X) \, dt
        +\int d^3k \, \bigg\{ \pi_t \big( \hat X \hat L(k) + \hat L (k)^\ast \hat X \big) \nonumber \\
        &&- \pi_t (\hat X) \pi_t \big( \hat L(k)+\hat L (k)^\ast \big) \bigg\}\, I(k,dt ),
        \label{eq:Stratonovich_Kushner}
    \end{eqnarray}
    where the innovations processes $I(k,t)$ are defined as
    \begin{eqnarray}
        \hat I(k, dt) = \hat Y(k,dt) - \pi_t \big( \hat L((k) + \hat L(k)^\ast \big) \, dt .
    \end{eqnarray}
\end{theorem}
\begin{proof}
    We shall adapt the characteristic function technique. To this end, we select an arbitrary function $f(k,t)$ and construct a stochastic process $t\mapsto \hat C(k,t)$ for each fixed $k$ by
    \begin{eqnarray}
        \hat C(k, dt ) = f(k,t) \, \hat C(t,k) \, Y(k,dt),
    \end{eqnarray}
    with $\hat C(k,0)=1$. We make the ansatz that the filter satisfies a stochastic differential equation of the form
    \begin{eqnarray}
        d\pi_t (\hat X ) = \alpha_t \, dt + \int d^3k\, 
        \beta_t (k) \hat Y (k, dt ) 
    \end{eqnarray}
    where $\alpha_t$ and $\beta_t (k)$ belong to $\mathfrak{Y}_t$.

    The conditional expectation has the property
    \begin{eqnarray}
        \left\langle \big[ \pi_t (\hat X) - j_t (\hat X)\big] \int d^3k' \, \hat C(k', t) \right\rangle=0
    \end{eqnarray}
    This takes the form $\langle S_tR_t\rangle=0$ and taking the time differential leads to $\langle dS_t \, R_t\rangle
    +\langle S_t\, dR_t\rangle+\langle dS_t\, dR_t\rangle =0$ by the Ito rule. We compute these three terms separately:
\begin{widetext}
\noindent        the first is
        \begin{gather}
        \left\langle \big[ d\pi_t (\hat X) - dj_t (\hat X)\big] \int d^3k' \, \hat C(k', t) \right\rangle
        =\left\langle \big[ \alpha_t + \int d^3 k \,  \beta_t (k) j_t \big( \hat L(k) +\hat L (k)^\ast \big) - j_t (\mathcal{L} \hat X ) \big] \int d^3k' \, \hat C(k', t) \right\rangle
        ;
    \end{gather}
        the second is
 \begin{gather}
        \left\langle \big[ \pi_t (\hat X) - j_t (\hat X)\big] \int d^3k' \, C(k', dt) \right\rangle
        =\left\langle \big[ \pi_t (\hat X) - j_t (\hat X) \big] \int d^3k' \, f(k',t) \, \hat C(k',t) \, Y(k',dt)\right\rangle
        ;
    \end{gather}
   while the third is
\begin{gather}
        \left\langle \big[ d\pi_t (\hat X) - d j_t (\hat X)\big] \int d^3k' \, C(k', dt) \right\rangle
        =\left\langle \int d^3k \, \beta (k,t)  f(k,t) \, \hat C(k,t) \right\rangle \, dt
                +
        \left\langle \int d^3k \,   f(k,t) \, \hat C(k,t) \, j_t \big( [ \hat L (k) , \hat X] \big) \right\rangle \, dt
        .
    \end{gather}
    \end{widetext}

    To obtain the third term, we used the Ito correction and the identities
    \begin{eqnarray}
        \widetilde{B}(k, dt) \, Y (k', dt) &=& \delta ^3 (k-k') dt, \nonumber \\
        \widetilde{B}(k, dt)^\ast \, Y (k', dt) &=& 0,
    \end{eqnarray}
    coming from the quantum Ito table.

    We now use the fact that the functions $f(k,t)$ were arbitrary and extract the coefficients of $ \hat C(k,t)$ and $f(k,t) \, \hat C(k,t)$, respectively, to obtain
    \begin{widetext}
        \begin{eqnarray}
        \pi_t \bigg( \big( \pi_t (\hat X) - j_t (\hat X )\big) j_t\big( \hat L (k) + \hat L(k)^\ast \big)\bigg)
        -\pi_t \big( [\hat L(k)^\ast , \hat X ]\big) +\beta_t (k)=0, \nonumber \\
        \alpha_t +\int d^3k \, \beta_t (k) \pi_t \big( L(k)+L(k)^\ast \big) - \pi_t (\mathcal{L} \hat X ) =0.
    \end{eqnarray}
    \end{widetext}
    
    Here we used the fact that $\pi_t (\alpha_t)= \alpha_t$ and $\pi_t (\beta_t(k))= \beta_t (k)$ as these processes already belong to $\mathfrak{Y}_t$, so that further coarse-graining is makes no change: $\pi_t ( j_t (\hat X))= \pi_t (\hat X)$, etc.

    We may now rearrange this to obtain (\ref{eq:Stratonovich_Kushner}).
\end{proof}

As is standard in filtering theory, the increments of the innovations are the observed (measured!) increments minus the expected increments. They have the statistics of a Gaussian field and will have the Ito rule
\begin{eqnarray}
    I(k,dt) \, I (k' , dt ) = \delta^3 (k-k') \, dt .
\end{eqnarray}

\subsection{Comparison with Di\'osi's Model}
The filter has been derived in the momentum representation, but is straightforward to return to the spatial representation. The corresponding innovations are given by
\begin{eqnarray}
    W(x,dt ) = \int d^3 k \frac{\sqrt{4 \pi \, G}}{|k|} \, e_k (x) \, I ( k , dt ) .
\end{eqnarray}

It follows that
\begin{eqnarray}
    \int d^3 k \, \hat L(k) \, I(k, dt ) = \int d^3 x \, \hat \mu (x) \, W(x, dt) ,
\end{eqnarray}
and we have the Ito rule
\begin{eqnarray}
    W(x, dt) \, W (x' , dt) = g (x,x' ) \, dt .
\end{eqnarray}

We therefore recover the Gaussian classical noise process introduced by Di\'osi with correlation (\ref{eq:DiosiCorr}).

The filter equation then becomes the stochastic differential equation
    \begin{eqnarray}
        d\pi_t (\hat X ) &=& \pi_t (\mathcal{L} \hat X) \, dt
        +\int d^3x \, \bigg\{ \pi_t \big( \hat X \hat \mu (x) + \hat \mu (x)^\ast \hat X \big) 
        \nonumber \\
        &&- 2\pi_t (\hat X) \pi_t \big( \hat \mu (x)  \big) \bigg\}\, W(x,dt ),
        \label{eq:Stratonovich_Kushner_Diosi}
    \end{eqnarray}

To make the connection with Di\'osi's equation, we set
\begin{eqnarray}
    \pi_t (\hat X ) = \langle \psi_t | \, \hat X | \, \psi_t \rangle.
\end{eqnarray}
Using (\ref{eq:DiosiSME}) for $|\psi_t \rangle$, we obtain (\ref{eq:Stratonovich_Kushner_Diosi}).

\subsection{Number Counting}
An alternative strategy is to measure the number of quanta for each $k$: this is the family of observables formally given by $Y(k,dt) = \tilde b_{\mathrm{out}}(k,t) \tilde b_{\mathrm{out}} (k,t)^\ast \, dt$. 
In analogy with the quantum optics situation, the quantum filter can be shown to satisfy
\begin{gather}
        d\pi_t (\hat X ) = \pi_t (\mathcal{L} \hat X) \, dt
        \nonumber \\
        +\int d^3k \, \bigg\{ \frac{\pi_t \big( \hat L (k)^\ast \hat X \hat L(k)\big)}{ \pi_t \big( \hat L (k)^\ast \hat L(k) \big)} - \pi_t (\hat X ) \big) \bigg\}\, I(k,dt ),
        \end{gather}
    where the innovations processes $I(k,t)$ are now compensated Poisson processes given by
    \begin{eqnarray}
        \hat I(k, dt) = \hat Y(k,dt) - \pi_t \big( \hat L (k)^\ast \hat L(k)  \big) \, dt .
    \end{eqnarray}
If we return to the space domain, we see that we no longer get a simply local form for the filtered observable in terms of the observables $\hat \mu (x)$.
%%%%%%%%%%%%%%%%%%%%%%%%%%%%%%%%%%%%%%%%%%%%%%%%%%%%%%
%%%%%%%%%%%%%%%%%%%%%%%%%%%%%%%%%%%%%%%%%%%%%%%%%%%%%%
%%%%%%%%%%%%%%%%%%%%%%%%%%%%%%%%%%%%%%%%%%%%%%%%%%%%%%
%%%%%%%%%%%%%%%%%%%%%%%%%%%%%%%%%%%%%%%%%%%%%%%%%%%%%%

\section{Discussion and Conclusion}
There is a view among physicists that they have a moral duty to produce a theory that unifies quantum theory and gravity. However, some researchers (notably Dyson) have argued that the two aspects cannot be reconciled in any meaningful way and that the Einsteinian theory of gravity may not be quantum. Additionally, the quantum measurement problem itself becomes problematic in any quantum gravity setting since the apparatus itself must be subject to the quantum fluctuations of space-time. As Dowker \cite{Dowker} has aptly pointed out \lq\lq \textit{The Copenhagen requirement for a physical measuring instrument in physical space external to the quantum system, is incompatible with the quantum system being spacetime.}\rq\rq

In the present paper we have looked at the Di\'osi-Penrose which deals with the simpler problem of Newtonian gravity, but they still address the main issue that superpositions of states, each having a massive particle localized in separate locations, ought to have a short decoherence time. 

We propose an alternative viewpoint here. Rather than treating space-time as a background source of fluctuations adding additional noise to any measurement apparatus, we would consider space-time itself as \textit{being} the measurement apparatus. The Di\'osi stochastic master equation then becomes (in the appropriate Newtonian gravitational approximation) the quantum filter that arises from space-time homodyning the continuum of output quadrature described in the open quantum stochastic model presented here. In short, the collapse of the wave-function occurs because the Universe is continuously observing its components. 

To be clear, we are not providing any specific mechanism at play here as to how the Universe functions as a detector. Much less do we explain how the measurement readout would tell space-time how to curve, though this is in spirit an analogue of quantum feedback applications from quantum optics. Instead, we are remarking that the Universe is large, macroscopic, and more-or-less classical (as far as General Relativity is concerned) and that quantum mechanics is not concerned with any underlying reality of microscopic systems but rather tells us what large macroscopic systems see when they measure them. But what we show, as a first step, is that the Di\'osi-Penrose theory can be formulated as a quantum filtering problem.

%%%%%%%%%%%%%%%%%%%%%%%%%%%%%%%%%%%%%%%%%%%%%%%%%%%%%%

%\begin{acknowledgments}
%We wish to acknowledge the support of the author community in using REV\TeX{}, offering suggestions and encouragement, testing new versions, \dots.
%\end{acknowledgments}

%\appendix

%%%%%%%%%%%%%%%%%%%%%%%%%%%%%%%%%%%%%%%%%%%%%%%%%%%%%%
%%%%%%%%%%%%%%%%%%%%%%%%%%%%%%%%%%%%%%%%%%%%%%%%%%%%%%


\begin{thebibliography}{99}
    \bibitem{Mott} N.F. Mott (1929), The Wave Mechanics of $\alpha$-Ray Tracks, Proceedings of the Royal Society of London. Series A, Containing Papers of a Mathematical and Physical Character. 126 (800)
    
    \bibitem{Joos_Zeh} Joos, E.; Zeh, H. D. (1985). "The emergence of classical properties through interaction with the environment". Zeitschrift für Physik B. 59 (2): 223–243
    
    \bibitem{Slava} Belavkin, V.P. (1989). A new wave equation for a continuous nondemolition measurement. Physics Letters A. 140 (7–8): 355–358
    
    \bibitem{BvH} L. Bouten, R. van Handel, M.R. James, An introduction to quantum filtering, SIAM J. Control Optim. 46, 2199-2241, (2007)
    
    \bibitem{Diosi87} Di\'osi, L. (1987). A universal master equation for the gravitational violation of quantum mechanics, Physics Letters A. 120 (8): 377–381.

    \bibitem{Diosi89} Di\'osi, L. (1989), Models for universal reduction of macroscopic quantum fluctuations, Physical Review A. 40 (3): 1165–1174

     \bibitem{Penrose} Penrose, R. (1996). On Gravity's role in Quantum State Reduction, General Relativity and Gravitation. 28 (5): 581–600.

     \bibitem{Lindblad} G. Lindblad, (1976) On the generators of quantum dynamical semigroups, Communications in Mathematical Physics. 48 (2): 119–130
     

    \bibitem{HP84} Hudson, R.L and K.R. Parthasarathy, Quantum Ito’s Formula and Stochastic Evolutions, Commun. Math. Phys., 93, 301-323 (1984).
    
    \bibitem{Partha} K.R. Parthasarathy, An Introduction to Quantum Stoch-astic Calculus, Monographs in Mathematics, Book 85, Birkhäuser, Basel (2012).

    \bibitem{GC85} C. W. Gardiner and M. J. Collett, Input and Output in Quantum Systems: Quantum Stochastic Differential Equations and Applications, Phys. Rev. A, 31, 3761-3774 (1985).

     \bibitem{RKHS} Berlinet, A., Thomas-Agnan, C.: Reproducing Kernel Hilbert Spaces in Probability and Statistics. Kluwer, Boston (2004).
     
    \bibitem{Dowker} F. Dowker,
    Unifying gravity and quantum theory requires better understanding of time, Nature 645, 32-34 (2025); see also arXiv:2601.02140.

\end{thebibliography}
\end{document}